\documentclass[runningheads]{llncs}
\usepackage{amsmath}
\usepackage{amsfonts}
\usepackage{amssymb}
\usepackage{mathptmx}

\newcommand{\eps}{\varepsilon}
\newcommand{\tto}[1][]{\mathrel{\vphantom{\xrightarrow{#1}}\smash{\xrightarrow{#1}}\vphantom{\to}^*}}

\begin{document}
\title{A Note on Limited Pushdown Alphabets in Stateless Deterministic Pushdown Automata}
\titlerunning{A Note on Limited Pushdown Alphabets in Stateless DPDAs}
\authorrunning{T. Masopust}
\author{Tom\'{a}\v{s} Masopust}
\institute{
  Institute of Mathematics, Academy of Sciences of the Czech Republic\\
  {\v Z}i{\v z}kova 22, 616 62 Brno, Czech Republic\\
  \email{masopust@math.cas.cz}
}

\maketitle
\begin{abstract}
  Recently, an infinite hierarchy of languages accepted by stateless deterministic pushdown automata has been established based on the number of pushdown symbols. However, the witness language for the $n$th level of the hierarchy is over an input alphabet with $2(n-1)$ elements. In this paper, we improve this result by showing that a binary alphabet is sufficient to establish this hierarchy. As a consequence of our construction, we solve the open problem formulated by Meduna et al. Then we extend these results to $m$\mbox{-}state realtime deterministic pushdown automata, for all $m\ge 1$. The existence of such a hierarchy for $m$-state deterministic pushdown automata is left open.
\end{abstract}

\section{Introduction}
  Pushdown automata, and especially deterministic pushdown automata, play a central role in many applications of context-free languages. Considering state complexity, it is well-known that any pushdown automaton can be transformed to an equivalent pushdown automaton with a single state. Thus, the hierarchy based on the number of states collapses to only one level. In that case, the sole state is not important and one-state automata are also called stateless in the literature. On the other hand, however, no such transformation is possible for deterministic pushdown automata since it is shown in~\cite{Harrison} that the number of states in deterministic pushdown automata establishes an infinite state hierarchy on the family of deterministic context-free languages. The reader is referred to~\cite{Harrison} for more details, including binary witness languages.
  
  Recently, Meduna et al.~\cite{mvz12} have established an infinite hierarchy of languages on the lowest level of the state hierarchy, on the level of stateless deterministic pushdown automata, based on the number of pushdown symbols. They have shown that there exist languages accepted by stateless deterministic pushdown automata with $n$ pushdown symbols that cannot be accepted by any stateless deterministic pushdown automaton with $n-1$ pushdown symbols. However, the size of the input alphabet of the witness language for each level of the hierarchy is almost twice bigger than the size of the pushdown alphabet, which is linear with respect to the level of the hierarchy. More specifically, they constructed an infinite sequence of languages $L_n$, for all $n\ge 1$, such that the alphabet of the language $L_n$ is of cardinality $2n$, and they proved that any stateless deterministic pushdown automaton accepting the language $L_n$ requires at least $n+1$ pushdown symbols. Thus, their construction requires a growing alphabet and it was left open whether the hierarchy can be shown with witness languages over a fixed alphabet. In addition, they formulated an open question of what is the role of non-input pushdown symbols in this hierarchy, that is, can a similar infinite hierarchy be established based on the number of non-input pushdown symbols?
  
  In this paper, we improve their result by constructing a sequence of witness languages over a binary alphabet, which is the best possible improvement because any unary language accepted by a stateless deterministic pushdown automaton is a singleton. In other words, languages accepted by deterministic pushdown automata (by empty pushdown) are prefix-free~\cite{Harrison}. As an immediate consequence of our construction, we obtain a solution to the open problem. Then we show that a similar idea can be used for realtime pushdown automata to establish infinite hierarchies based on the number of pushdown symbols on each level of the state hierarchy. Again, the witness languages are binary. However, the case of general deterministic pushdown automata is left open.

  Stateless automata of many types are of great interest and have widely been investigated in the literature. For instance, stateless multicounter machines have been studied in~\cite{ei09}, hierarchies of stateless multicounter $5'\to 3'$ Watson-Crick automata have been investigated in~\cite{ehn11,nhe11}, multihead finite and pushdown automata have been studied in~\cite{fi09,iko08}, stateless restarting automata in~\cite{kmo10,kmo10b}, and stateless automata and their relationship to P systems in~\cite{ydi08}. Pushdown store languages, that is, languages consisting of strings occurring on the pushdown along accepting computations of a pushdown automaton, have been studied in~\cite{malcher,malcherTR}, including the languages of stateless pushdown automata. Some decision results concerning stateless and deterministic pushdown automata can be found in~\cite{Gallier81,Valiant}.

\section{Preliminaries}
  In this paper, we assume that the reader is familiar with automata and formal language theory~\cite{Harrison,Hopcroft:1969:FLR:1096945,Salomaa}. For a set $A$, $|A|$ denotes the cardinality of $A$, and $2^A$ denotes the powerset of $A$. An alphabet is a finite nonempty set. For an alphabet $V$, $V^*$ represents the free monoid generated by $V$, where the unit (the empty string) is denoted by $\eps$.
  
  A {\em pushdown automaton\/} (PDA) is a septuple $M = (Q, \Sigma, \Gamma, \delta, q_0, Z_0, F)$, where $Q$ is the finite set of states, $\Sigma$ is the input alphabet, $\Gamma$ is the pushdown alphabet, $\delta$ is the transition function from $Q\times \Gamma \times (\Sigma\cup\{\eps\})$ to the set of finite subsets of $Q\times \Gamma^*$, $q_0\in Q$ is the initial state, $Z_0\in\Gamma$ is the initial pushdown symbol, and $F\subseteq Q$ is the set of accepting states. A configuration of $M$ is any element of $Q\times \Gamma^*\times\Sigma^*$. For a configuration $(q,\gamma,\sigma)$, $q$ denotes the current state of $M$, $\gamma$ denotes the content of the pushdown (with the top as the rightmost symbol), and $\sigma$ denotes the unread part of the input string. The automaton changes the configuration according to the transition function $\delta$. This one-step relation is denoted by $\vdash$ and defined so that $(q,\gamma A, a\sigma)\vdash (p,\gamma w, \sigma)$ if $(p,w)$ is in $\delta(q,A,a)$, where $\gamma$ and $w$ are strings of pushdown symbols, $A$ is a pushdown symbol, $a$ is an input symbol or $\eps$, and $\sigma$ is a string of input symbols. Let $\vdash^*$ denote the reflexive and transitive closure of the relation $\vdash$. The language accepted by $M$ (by a final state and empty pushdown) is denoted by $L(M)$ and defined as $L(M) = \{ w \in\Sigma^* \mid (q_0, Z_0, w) \vdash^* (f, \eps,\eps)\text{ for some } f\in F\}$. 
  
  Automaton $M$ is {\em deterministic} (DPDA) if there is no more than one move the automaton can make from any configuration, that is, $|\delta(q,A,a)|\le 1$, for any state $q$, pushdown symbol $A$, and input symbol $a$, and if $\delta(q,A,\eps)$ is defined, that is, $\delta(q,A,\eps)\neq\emptyset$, then $\delta(q,A,a)$ is not defined for any input symbol $a$. 
  
  Automaton $M$ is {\em realtime} (RPDA) if it is deterministic and if no $\eps$-transitions are defined, that is, if $\delta(q,A,a)$ is defined, then $a\neq\eps$.
    
  Automaton $M$ is {\em stateless\/} (SPDA) if it has only one state, that is, $|Q|=1$. Moreover, in stateless pushdown automata, we allow an initial pushdown string $\alpha\in\Gamma^+$ instead of an initial pushdown symbol. Thus, in this case, we simply write $M = (\Sigma, \Gamma, \delta, \alpha)$. Note that unlike the general case of deterministic pushdown automata, in the case of stateless deterministic pushdown automata, the realtime restriction does not decrease the power of the machine~\cite{Harrison}.

  In what follows, the notation $(q,\gamma) \xrightarrow{a} (p,\gamma')$, for $a$ in $\Sigma$ and $\gamma,\gamma'$ in $\Gamma^*$, denotes the computational step $(q,\gamma,a)\vdash (p,\gamma',\eps)$, and for a string $w\in\Sigma^*$, the notation $(q,\gamma) \tto[w] (p,\gamma')$ denotes the maximal computation $(q,\gamma, w)\vdash^* (p, \gamma',\eps)$, that is, no other $\eps$-tran\-si\-tions are possible from the configuration $(p, \gamma',\eps)$. For stateless automata, the state is eliminated from this notation.
  
  We now recall the definition of families of languages accepted by stateless deterministic pushdown automata with $n$ pushdown symbols.
  \begin{definition}\label{def1}
    A stateless deterministic pushdown automaton $M=(\Sigma,\Gamma,\delta,\alpha)$ is $n$-pushdown-alphabet-limited, for $n\ge 1$, if $|\Gamma|\le n$. The family $n$-SDPDA denotes all languages accepted by an $n$-pushdown-alphabet-limited stateless deterministic pushdown automaton. Similarly, the family $n$-SRPDA denotes all languages accepted by an $n$-pushdown-alphabet-limited stateless realtime pushdown automaton.
  \end{definition}

\section{Main results}
  Recall that it is known~\cite{Harrison} that any language $L$ accepted by a stateless deterministic pushdown automaton is prefix-free. Thus, if $L$ contains $\eps$, then $L=\{\eps\}$. It is also known that any stateless deterministic pushdown automaton accepting a language different from $\{\eps\}$ can be transformed to an equivalent stateless realtime pushdown automaton. To simplify proves, we extend this result by showing that the resulting stateless realtime pushdown automaton does not require any new pushdown symbol.
  
  \begin{theorem}\label{thm1}
    For any stateless deterministic pushdown automaton with $n$ pushdown symbols accepting a language $L\neq\{\eps\}$, there exists an equivalent stateless realtime pushdown automaton with no more than $n$ pushdown symbols.
  \end{theorem}
  \begin{proof}
    Assume that all pushdown symbols appear in an accepting computation. Let $X\xrightarrow{\ \eps\ }\sigma$ be an $\eps$-transition of $M=(\Sigma,\Gamma,\delta,\alpha)$. If $X$ appears in $\sigma$, then any computation reaching $X$ on top of the pushdown is not accepting because $X$ cannot be eliminated from the pushdown; it is always replaced with $\sigma$ containing $X$ again, which contradicts our assumption. Modify the automaton $M$ by replacing all occurrences of $X$ on right-hand sides of transitions with $\sigma$, and by removing $X\xrightarrow{\ \eps\ }\sigma$ from $\delta$ and $X$ from $\Gamma$. The accepted language is not changed because the only action of $X$ was to be replaced with $\sigma$, and the resulting automaton has one $\eps$-transition less. Repeating this substitution results in an equivalent stateless realtime pushdown automaton with no more than $n$ pushdown symbols.
  \end{proof}

  This result can be rewritten as follows.
  \begin{corollary}
    For all $n\ge 1$, $n\text{-SDPDA} =n\text{-SRPDA}$.
  \end{corollary}

  In the following subsection, we improve the result presented in~\cite{mvz12} by giving binary witness languages and answering the open problem formulated there.

\subsection{Binary alphabet}
  First, note that from the fact that any language accepted by a stateless deterministic pushdown automaton is prefix-free, any such a unary language is of the form $a^n$, for some $n\ge 0$. Every such language can be accepted by a stateless deterministic pushdown automaton with only one pushdown symbol (equal to the input symbol) starting from the initial string $a^n$ with $\delta(a,a)=\eps$. Thus, to establish the infinite hierarchy based on the number of pushdown symbols using only a binary alphabet is the best possible improvement that can be done. To do this, we define the following sequence of languages
  \begin{equation}\label{LN}
    \begin{aligned}
      L_1 & = \{a^c\}, \text{ for some } c\ge 0\,,\text{ and }\\
      L_n & = \{b^k a \mid 1\le k\le n-1\}, \text{ for any } n\ge 2\,.
    \end{aligned}
  \end{equation}
  
  We now prove that for all $n\ge 1$, any stateless deterministic pushdown automaton accepting the language $L_n$ requires at least $n$ pushdown symbols. First, we show that $n$ pushdown symbols is sufficient.
  
  \begin{lemma}
    For any $n\ge 1$, there exists a stateless deterministic pushdown automaton with $n$ pushdown symbols accepting the language $L_n$ defined in (\ref{LN}).
  \end{lemma}
    \begin{proof}
    As mentioned above, there exists a stateless deterministic pushdown automaton with one pushdown symbol accepting the language $L_1$. Thus, consider the language $L_n$ for $n\ge 2$. This language is accepted by the stateless pushdown automaton $M_n=(\{a,b\},\{X_0,X_1,\ldots,X_{n-1}),\delta,X_0)$, where the transition function $\delta$ is defined as follows:
    \begin{enumerate}
      \item $\delta(X_i,b)=X_{i+1}$, for $0\le i\le n-2$,
      \item $\delta(X_i,a)=\eps$, for $1\le i \le n-1$.
    \end{enumerate}
    Symbols $X_1, X_2,\ldots, X_{n}$ are used to count the number of $b$'s and to restrict this number to $n-1$. Notice that $M_n$ is deterministic (realtime) and that $|\Gamma|=n$.
  \end{proof}
    
  We next prove that at least $n$ pushdown symbols are necessary for the language $L_n$.
  \begin{lemma}
    For any $n\ge 1$, there does not exist a stateless deterministic pushdown automaton accepting the language $L_n$ defined in (\ref{LN}) with $n-1$ pushdown symbols.
  \end{lemma}
  \begin{proof}
    The case $n=1$ is trivial and follows from the definition. Thus, consider the case $n\ge 2$. Let $M_n=(\{a,b\},\Gamma,\delta,\alpha)$ be a stateless deterministic pushdown automaton accepting the language $L_n$. By Theorem~\ref{thm1}, we can assume that $M_n$ is realtime. Consider the computations on strings $b^ka$ for all $k=1,2,\ldots,n-1$. Then
    \[
      \alpha \xrightarrow{\ b\ } \pi_1 X_1\xrightarrow{\ b\ }\pi_2 X_2 \xrightarrow{\ b\ }\ldots\xrightarrow{\ b\ } \pi_{n-1} X_{n-1} \tto[\ a\ ]\eps\,.
    \]
    As symbol $a$ is accepted from all configurations $\pi_k X_k$, it implies that $X_k\xrightarrow{\ a\ }\eps$ and $\pi_k=\eps$, for all $k=1,2,\ldots,n-1$, that is,
    \[
      \alpha \xrightarrow{\ b\ } X_1\xrightarrow{\ b\ } X_2 \xrightarrow{\ b\ }\ldots\xrightarrow{\ b\ } X_{n-1} \xrightarrow{\ a\ }\eps\,.
    \]
    Assume that there exist $i<j$ such that $X_i=X_j$. Since $X_i\tto[\ b^{n-1-i}a\ ]\eps$, the automaton accepts the string $b^{j+n-1-i}a$ that does not belong to the language $L_n$ because $n-1+(j-i)>n-1$. Hence, all the pushdown symbols $X_i$, for $i=1,2,\ldots,n-1$, are pairwise different. It remains to show that the automaton needs at least one more pushdown symbol. To do this, consider the form of the initial string $\alpha$. If $\alpha\in\{X_1,X_2,\ldots,X_{n-1}\}^+$, then $\alpha\tto[\ a^{|\alpha|}\ ] \eps$ because $X_i\xrightarrow{\ a\ }\eps$, for all $i=1,2,\ldots,n-1$. Thus, the automaton needs at least $n$ pushdown symbols.
  \end{proof}

  As a consequence of the previous two lemmas, we have the following result. The strictness of the inclusion has been shown in~\cite{mvz12}. We have added binary witness languages.
  \begin{theorem}
    For each $n\ge 1$, $n$-SDPDA $\subsetneq (n+1)$-SDPDA. In addition, the witness languages are binary.
  \end{theorem}

\subsection{Non-input pushdown symbols}
  In~\cite{mvz12}, the authors formulated a question of whether an infinite hierarchy can be established based on the number of non-input pushdown symbols, that is, whether there exists a similar hierarchy when Definition~\ref{def1} is modified as follows.
  \begin{definition}\label{def2}
    A stateless deterministic pushdown automaton $M=(\Sigma,\Gamma,\delta,\alpha)$ is an $n$-non-input-pushdown-alphabet-limited, for $n\ge 1$, if the number of non-input pushdown symbols is limited by $n$, that is, if $|\Gamma|-|\Sigma| \le n$. The family $n$-niSDPDA denotes all languages accepted by an $n$-non-input-pushdown-alphabet-limited stateless deterministic pushdown automaton.
  \end{definition}
  
  As an immediate consequence of the previous results, we get the answer to this question.
  \begin{theorem}
    For each $n\ge 0$, $n$-niSDPDA $\subsetneq (n+1)$-niSDPDA.
  \end{theorem}
  \begin{proof}
    We define languages $K_n = L_{n+2}$. Then any stateless deterministic pushdown automaton accepting the language $K_n$ requires $n+2$ pushdown symbols. Since the input alphabet is binary, we immediately get $|\Gamma|-|\Sigma|=n$, as required.
  \end{proof}

\section{Generalization to realtime pushdown automata}
  It is a natural question to ask whether a similar hierarchy can be obtained for the other levels of the state hierarchy. To prove this for realtime pushdown automata, we define the following sequence of languages
  \begin{equation}\label{MN}
    \begin{aligned}
      L_{m,1} & = \{a^c\}, \text{ for some } c\ge 0\,,\text{ and }\\
      L_{m,n} & = \{b^k a \mid 1\le k\le mn-1\}, \text{ for any } m\ge 1 \text{ and } n\ge 2\,.
    \end{aligned}
  \end{equation}
  
  We now prove that for all $m,n\ge 1$, any $m$-state deterministic pushdown automaton accepting the language $L_{m,n}$ requires at least $n$ pushdown symbols. The sufficiency is shown first. 
  
  \begin{lemma}
    For any $m,n\ge 1$, there exists an $m$-state realtime pushdown automaton with $n$ pushdown symbols accepting the language $L_{m,n}$ defined in (\ref{MN}).
  \end{lemma}
  \begin{proof}
    Consider the language $L_{m,n}$ for $m\ge 1$ and $n\ge 2$. This language is accepted by the $m$-state deterministic pushdown automaton \[M_{m,n}=(\{q_0,q_1,\ldots,q_{m-1}\},\{a,b\},\{X_0,X_1,\ldots,X_{n-1}),\delta,q_0,X_0,\{q_{m-1}\})\,,\] where the transition function $\delta$ is defined as follows:
    \begin{enumerate}
      \item $\delta(q_j,X_i,b)=(q_j,X_{i+1})$, for $0\le j\le m-1$ and $0\le i\le n-2$,
      \item $\delta(q_j,X_{n-1},b)=(q_{j+1},X_0)$, for $0\le j\le m-2$,
      \item $\delta(q_0,X_i,a)=(q_{m-1},\eps)$, for $1\le i\le n-1$,
      \item $\delta(q_j,X_i,a)=(q_{m-1},\eps)$, for $1\le j\le m-1$ and $0\le i\le n-1$.
    \end{enumerate}
    Pushdown symbols in combination with states are used to count the number of $b$'s and to restrict this number to $mn-1$. Notice that $M_{m,n}$ is realtime and that $|\Gamma|=n$.
  \end{proof}
    
  We now prove that at least $n$ pushdown symbols are necessary for any $m$-state realtime pushdown automaton to accept the language $L_{m,n}$.
  \begin{lemma}
    For any $m,n\ge 1$, there does not exist any $m$-state realtime pushdown automaton accepting the language $L_{m,n}$ defined in (\ref{MN}) with $n-1$ pushdown symbols.
  \end{lemma}
  \begin{proof}
    Let $m\ge 1$ be fixed, but arbitrary. The case $n=1$ is trivial and follows from the definition. Thus, consider the case $n\ge 2$. For the sake of contradiction, assume that there exists an $m$-state realtime pushdown automaton $M_{m,n}=(\{q_0,q_1,\ldots,q_{m-1}\},\{a,b\},\Gamma,\delta,$ $q_0,X_0,F)$ with $n-1$ pushdown symbols accepting the language $L_{m,n}$. Consider the computations on strings $b^ka$ for all $k=1,2,\ldots,mn-1$. Then
    \begin{align*}
      (q_0,X_0) & \xrightarrow{\ b\ } (q_1,\pi_1 X_1)\xrightarrow{\ b\ }(q_2,\pi_2 X_2) \xrightarrow{\ b\ }\ldots\\
        \ldots  & \xrightarrow{\ b\ } (q_{mn-1},\pi_{mn-1} X_{mn-1}) \tto[\ a\ ](q_f,\eps)\,,
    \end{align*}
    where $q_f\in F$. As symbol $a$ is accepted from all configurations $(q_k,\pi_k X_k)$, it implies that $(q_k,X_k)\xrightarrow{\ a\ }(q_{f_k},\eps)$, for some $q_{f_k}\in F$, and $\pi_k=\eps$, for all $k=1,2,\ldots,mn-1$, that is,
    \[
      (q_0,X_0) \xrightarrow{\ b\ } (q_1,X_1)\xrightarrow{\ b\ }(q_2,X_2) \xrightarrow{\ b\ }\ldots\xrightarrow{\ b\ } (q_{mn-1},X_{mn-1}) \xrightarrow{\ a\ }(q_f,\eps)\,.
    \]
    As we have $m$ states and $n-1$ pushdown symbols, there must exist $i<j$ such that $(q_i,X_i)=(q_j,X_j)$. Since $(q_i,X_i)\tto[\ b^{mn-1-i}a\ ]\eps$, the automaton accepts string $b^{j+mn-1-i}a$ that does not belong to the language $L_{m,n}$ because $mn-1+(j-i)>mn-1$, which is a contradiction. Thus, the automaton needs at least $n$ pushdown symbols.
  \end{proof}
  
  We now present definitions of families of languages accepted by $m$-state deterministic pushdown automata with $n$ pushdown symbols.
  \begin{definition}
    A pushdown automaton $M=(Q,\Sigma,\Gamma,\delta,q_0,Z_0,F)$ is $n$-pushdown-alpha\-bet-limited, for $n\ge 1$, if $|\Gamma|\le n$. The family $(m,n)$-RPDA denotes all languages accepted by an $m$-state realtime pushdown automaton with $n$ pushdown symbols.
  \end{definition}
  \begin{definition}
    A pushdown automaton $M=(Q,\Sigma,\Gamma,\delta,q_0,Z_0,F)$ is an $n$-non-input-pushdown-alphabet-limited, for $n\ge 1$, if $|\Gamma|-|\Sigma| \le n$. The family $(m,n)$\mbox{-}niRPDA denotes all languages accepted by an $m$-state realtime pushdown automaton with $n$ non-input pushdown symbols.
  \end{definition}

  As a consequence of the previous two lemmas, we have the following results.
  \begin{theorem}
    For each $m,n\ge 1$, $(m,n)$-RPDA $\subsetneq (m,n+1)$-RPDA. In addition, the witness languages are binary.
  \end{theorem}
  
  \begin{theorem}
    For each $m\ge 1$ and $n\ge 0$, $(m,n)$-niRPDA $\subsetneq (m,n+1)$-niRPDA.
  \end{theorem}
  \begin{proof}
    We define languages $K_{m,n} = L_{m,n+2}$. Then any $m$-state realtime pushdown automaton accepting the language $K_{m,n}$ requires $n+2$ pushdown symbols. Since the input alphabet is binary, we immediately get $|\Gamma|-|\Sigma|=n$, as required.
  \end{proof}

  Finally, note that Theorem~\ref{thm1} does not hold for deterministic pushdown automata that are not stateless. The following example shows that there exists a two-state deterministic automaton with two pushdown symbols accepting the language $L_{m,n}$, for all $m,n\ge 1$. This implies that to establish a similar infinite hierarchy for deterministic pushdown automata based on the number of states and pushdown symbols over a binary (constant) input alphabet requires more sophisticated approaches.
  \begin{example}
    Let $m,n\ge 1$ be fixed, but arbitrary. Then the automaton $M=(\{f,q\},\{a,b\},$ $\{Z,B\},\delta,f,Z,\{f\})$, where $\delta$ is defined so that
    \begin{enumerate}
      \item $\delta(f,Z,b)=(q,B^{mn-1})$,
      \item $\delta(q,B,b)=(q,\eps)$,
      \item $\delta(q,B,a)=(f,\eps)$,
      \item $\delta(f,B,\eps)=(f,\eps)$,
    \end{enumerate}
    accepts the language $L_{m,n}$. The automaton reads the first input symbol $b$, changes its state, and pushes a string of $mn-1$ symbols $B$ to the pushdown. Then it compares input symbol $b$ against $B$ on the pushdown until it reads symbol $a$ with $B$ on the pushdown. In that case, it goes to the other state where it empties the whole pushdown without reading any other input symbol.
  \end{example}

\section{Conclusion}
  It is a natural question to ask whether a similar hierarchy can be obtained for deterministic pushdown automata. It is likely that such a hierarchy exists on each level, however, we have not established it in this paper. Can such a hierarchy be established with witness languages over a binary (constant) alphabet?

  Languages accepted by stateless deterministic pushdown automata are also called {\em simple languages\/} because they are generated by so-called {\em simple grammars}~\cite{Harrison}. It was shown in~\cite{KorenjakHopcroft} that the equivalence problem for simple grammars is decidable, which is, together with the general result by S{\'e}nizergues~\cite{Senizergues99,Senizergues2001}, in contrast to the fact that containment is undecidable even for simple languages as shown in~\cite{Friedman}. It is remarkable that the best known algorithm for the equivalence of simple languages works in time $O(n^6 \text{ polylog } n)$, where $n$ is the size of the grammar, see~\cite{LasotaRytter}. Hence, simple languages are in some sense the simplest languages for which containment is undecidable. Note that further simplification results in so-called {\em very simple languages\/}, for which the containment was shown decidable in~\cite{WT93}, and the algorithm improved and simplified in~\cite{Makinen}. Simple languages also play a role in process algebras, see, e.g.,~\cite{GrooteH94}.
  
  Finally, note that it is unsolvable to decide whether for a context-free language $L$ there exists an integer $n$ such that $L$ is accepted by a stateless deterministic pushdown automaton with a pushdown alphabet of cardinality $n$. This follows immediately from the undecidability of the problem whether a context-free language $L$ can be accepted by a stateless deterministic pushdown automaton, see~\cite{Harrison} where this problem is formulated as an exercise. As far as the author knows, it is also an open problem whether there exists an algorithm to compute, for a given deterministic context-free language $L$, the minimal $n$ such that $L$ belongs to the $n$th level of the state hierarchy.

\section*{Acknowledgement.}
  Research supported by the GA{\v C}R grant no. P202/11/P028 and by RVO:~67985840.

\bibliographystyle{plain}
\bibliography{masopust}

\end{document}